\documentclass[11pt,letterpaper,abstracton]{scrartcl}

\usepackage[margin=1in]{geometry}
\usepackage{authblk}

\usepackage{amsopn}
\usepackage{amsmath}
\usepackage{amsthm}
\usepackage{amssymb}
\usepackage{mathabx}

\usepackage{complexity}
\usepackage{environ}
\usepackage{comment}
\usepackage{hyperref}
\usepackage{xcolor,colortbl}
\usepackage{tabularx}
\usepackage{wrapfig}

\newtheorem {theorem}{Theorem}

\newtheorem {lemma}[theorem]{Lemma}
\newtheorem {proposition}[theorem]{Proposition}

\newtheorem {fact}[theorem]{Fact}

\newcommand{\F}{\mathbb{F}_2}
\newcommand{\N}{\mathbb{N}}
\newcommand{\Z}{\mathbb{Z}}
\newcommand{\powerset}[1] {\mathcal{P}(#1)}
\newcommand{\mylog}{\mathit{log}}
\newcommand{\set}[1]{\{#1\}}
\newcommand{\Set}[2]{\{#1\ \mid\, #2\}}
\newcommand{\power}[1]{|#1|}
\newcommand{\factorize}[2]{#1/_{#2}}
\newcommand{\charfun}[1]{\chi_{#1}}

\newcommand{\witcountprob}{\ComplexityFont{Witness\ Counting}}
\newcommand{\aset}{S}
\newcommand{\setvec}{V}
\newcommand{\sizesetvec}{m}
\newcommand{\adim}{d}
\newcommand{\anum}{k}
\newcommand{\allvec}{\F^{\adim}}
\newcommand{\avec}{v}
\newcommand{\avecvec}{\vec v}
\newcommand{\avecvecp}{\vec u}
\newcommand{\avecvecpp}{\vec w}
\newcommand{\tvec}{t}
\newcommand{\ops}{\mathit{ops}}

\newcommand{\witof}[3]{\mathit{Wit}(#1, #2, #3)}
\newcommand{\candof}[3]{\mathit{Cand}(#1, #2, #3)}
\newcommand{\failof}[3]{\mathit{Fail}(#1, #2, #3)}
\newcommand{\witgiven}{\witof{\setvec}{\tvec}{\anum}}
\newcommand{\candgiven}{\candof{\setvec}{\tvec}{\anum}}
\newcommand{\failgiven}{\failof{\setvec}{\tvec}{\anum}}

\newcommand{\hadmatrix}{H}
\newcommand{\WHT}{\mathit{WHT}}

\newcommand{\shape}{\mathit{usage}}
\newcommand{\shapeinv}{\shape^{-1}}
\newcommand{\shapeof}[1]{\shape(#1)}
\newcommand{\shapeinvof}[1]{\shapeinv(#1)}
\newcommand{\inst}{\mathit{inst}}
\newcommand{\instof}[1]{\inst(#1)}
\newcommand{\setinst}{\mathit{Inst}}
\newcommand{\setinstof}[1]{\setinst(#1)}
\newcommand{\aupart}{\mathit{up}}
\newcommand{\upartof}[1]{\mathit{UPart}(#1)}
\newcommand{\vecequivof}[1]{\equiv_{#1}}
\newcommand{\parityfun}{\mathit{paritycount}}
\newcommand{\parityfunof}[1]{\parityfun(#1)}
\newcommand{\parityfuninv}{\parityfun^{-1}}
\newcommand{\parityfuninvof}[1]{\parityfuninv(#1)}

\newcommand{\oddcount}{\mathit{o}}
\newcommand{\evencount}{\mathit{e}}
\newcommand{\evenof}[1]{\mathit{evencount}(#1)}
\newcommand{\oddof}[1]{\mathit{oddcount}(#1)}

\newcommand{\graph}{\mathit{G}}
\newcommand{\univ}{\mathit{U}}
\newcommand{\edges}{\mathit{E}}
\newcommand{\per}{\mathit{P}}
\newcommand{\perfect}[1]{\mathit{Perf}(#1)}
\newcommand{\alg}{\mathcal{A}}

\newcommand{\bigOS}[1]{\mathcal{O}^*(#1)}
\newcommand{\bigO}[1]{\mathcal{O}(#1)}
\newcommand{\smallo}[1]{{o}(#1)}
\newcommand{\ETH}{\ComplexityFont{ETH}}

\newcommand{\multtime}[1]{\mathit{mtime}(#1)}
\newcommand{\parcounttab}[3]{\mathit{pctab}_{#1}(#2, #3)}
\newcommand{\steiner}{\ComplexityFont{Steiner~Tree}}
\newcommand{\domset}{\ComplexityFont{Dominating~Set}}
\newcommand{\bandwidth}{\ComplexityFont{Bandwidth}}
\newcommand{\mld}{\ComplexityFont{Maximum~Likelihood~Decoding}}
\newcommand{\weightdistr}{\ComplexityFont{Weight~Distribution}}
\newcommand{\mindist}{\ComplexityFont{Minimum~Distance}}
\newcommand{\syndrome}{\ComplexityFont{Syndrome~Decoding}}
\newcommand{\kpath}[1]{#1\text{-}\ComplexityFont{Path}}
\newcommand{\subsum}{\ComplexityFont{Subset~Sum}}
\newcommand{\perfmatch}[1]{\ComplexityFont{Perfect~#1\text{-}Matching}}
\newcommand{\conpoly}{\ComplexityFont{coNP}/\poly}

\makeatletter
\newcommand {\problemtitle}[1]{\gdef\@problemtitle{#1}}
\newcommand {\problemshort}[1]{\gdef\@problemshort{#1}}
\newcommand {\probleminput}[1]{\gdef\@probleminput{#1}}
\newcommand {\problemparameter}[1]{\gdef\@problemparameter{#1}}
\newcommand {\problemquestion}[1]{\gdef\@problemquestion{#1}}

\NewEnviron{problem}{
	\problemtitle{}
	\problemshort{}
	\probleminput{}
	\problemquestion{}
	\BODY
	\par\addvspace{.5\baselineskip}
	\noindent
	\framebox[\textwidth]{
		\begin{tabularx}{\textwidth}{@{\hspace{\parindent}} l X c}
			\multicolumn{2}{@{\hspace{\parindent}}l}{ \normalsize \emph{\@problemtitle} \@problemshort} \\
			\normalsize \textbf{Input:} & \normalsize \@probleminput \\
			\normalsize \textbf{Task:} & \normalsize \@problemquestion
		\end{tabularx}
	}
	\par\addvspace{.5\baselineskip}
}

\begin{document}

	\title{Fast Witness Counting}
	\author[1]{Peter Chini}
	\author[2]{Rehab Massoud}
	\author[1]{Roland Meyer}
	\author[1]{Prakash Saivasan}
	\affil[1]{TU Braunschweig, \texttt{\{p.chini, roland.meyer, p.saivasan\}@tu-bs.de}}
	\affil[2]{University of Bremen, \texttt{massoud@uni-bremen.de}}
	\date{}

	\maketitle
	\vspace{-1.25cm}
	
	\begin{abstract}
	We study the witness-counting problem: given a set of vectors $\setvec$ in the $\adim$-dimensional vector space over $\F$, a target vector $\tvec$, and an integer $\anum$, count all ways to sum-up exactly $\anum$ different vectors from $\setvec$ to reach $\tvec$.
	The problem is well-known in coding theory and received considerable attention in complexity theory.
	Recently, it appeared in the context of hardware monitoring.
	
	Our contribution is an algorithm for witness counting that is optimal in the sense of fine-grained complexity. 
	It runs in time $\bigOS{2^{\adim}}$ with only a logarithmic dependence on $\sizesetvec=\power{\setvec}$. 
	The algorithm makes use of the Walsh-Hadamard transform to compute convolutions over $\F^d$. 
	The transform, however, overcounts the solutions.
	Inspired by the inclusion-exclusion principle, we introduce correction terms.
	The correction leads to a recurrence that we show how to solve efficiently.
	The correction terms are obtained from equivalence relations over $\F^d$.
	
	We complement our upper bound with two lower bounds on the problem.
	The first relies on $\# \ETH$ and prohibits an $2^{o(\adim)}$-time algorithm.
	The second bound states the non-existence of a polynomial kernel for the decision version of the problem.
\end{abstract}
	
	\section{Introduction}
\label{Section:Introduction}

	We address the \emph{witness-counting problem} (\witcountprob) defined as follows.
	Given a finite set of $\adim$-dimensional vectors $\setvec\subseteq\allvec$ over the field of characteristic two, a target vector $\tvec\in\allvec$, and a number $\anum\in\N$, determine $\power{\witgiven}$ with
	\begin{align*}
		\witgiven = \Set{\avecvec=(\avec_1,\ldots, \avec_{\anum})\in\setvec^{\anum}}{\Sigma \avec_i = \tvec, \avec_i\neq \avec_j\text{ for all } i\neq j}.
	\end{align*}
	The set consists of so-called \emph{witnesses}, $\anum$-tuples of pairwise-distinct vectors in $\setvec$ that sum up to $\tvec$.
	
	\witcountprob\ generalizes well-known algorithmic problems from coding theory.
	Prominent examples of such are $\mld$, $\syndrome$, $\mindist$, and $\weightdistr$.
	These problems arise from a decoding task.
	Consider a word received from a noisy channel.
	Due to the noise, the word may contain errors and differ from the codeword that was actually sent. 
	As a receiver, we are interested in the actual codeword, and it is our task to reconstruct it. 
	Usually, the number of errors in the received word is bounded by a measure called the \emph{Hamming weight}.
	Hence, we need to decide the existence of a codeword close enough to the received word wrt. the given Hamming weight.
	
	All four problems mentioned above are variants of $\witcountprob$, where the target vector is zero or solutions are allowed to contain at most $k$ vectors.
	The Hamming weight is always modeled by the parameter $k$.
	These problems are algorithmically hard.
	A series of papers \cite{Berlekamp1978,Bruck1990,Vardy1997,Downey1999,Vardy2005} studies their complexity and shows results ranging from $\NP$-completeness of all problems to $\W[1]$-hardness of $\mld$ and $\weightdistr$.
	Vardy provides a survey~\cite{VardySurvey1997}.
	
	From an algorithmic perspective, much effort was invested into finding randomized procedures for decoding. 
	One of the first decoding algorithms was introduced by Prange in \cite{Prange1962}.
	The author 
	\begin{wraptable}{r}{5.5cm}
		\vspace{-0.2cm}
		\centering
		\begin{tabular}{ |m{2.2cm} | m{1.8cm}|}
			\hline
			\scriptsize{Authors} & \scriptsize{Runtime} \\ 
			\hline 
			\hline
			\scriptsize{Prange \cite{Prange1962}} & \scriptsize{$\bigOS{2^{1/17 \sizesetvec}}$} \\
			\hline
			\scriptsize{Stern \cite{Stern1988}} & \scriptsize{$\bigOS{2^{1/18 \sizesetvec}}$} \\
			\hline
			\scriptsize{May et al. \cite{May2011}} & \scriptsize{$\bigOS{2^{1/19 \sizesetvec}}$} \\
			\hline
			\scriptsize{Becker et al. \cite{Becker2012}} & \scriptsize{$\bigOS{2^{1/20 \sizesetvec}}$} \\
			\hline
		\end{tabular}
	\end{wraptable}
	proposes a technique called \emph{Information Set Decoding} (ISD).
	It uses linear algebra (random permutations) to reduce the search space of potential codewords.
	Since then, ISD was combined with other search techniques,
	prominently the representation technique for $\subsum$ from \cite{Joux2010}. 
	This led to improved runtimes, an overview of which is given in the table to the right. 
	To be precise, the table refers to the decision version of $\witcountprob$, checking whether $\power{\witgiven}>0$. 
	Due to the randomization, these algorithms are not suitable for witness counting.
	Moreover, all runtimes depend exponentially on $\sizesetvec$.
	This means they are intractable on instances where the set of vectors $\setvec$ and hence $\sizesetvec$ tends to be large.

	Such instances arise in the context of a new logging procedure in hardware monitoring~\cite{Massoud2017}.
	There, a signal is traced on an interval of $\sizesetvec$ clock cycles.
	Each clock cycle is assigned a bitvector in $\F^\adim$ uniquely identifying it.
	The vectors get collected in the set $\setvec$.
	Roughly, the logging procedure adds up all vectors of clock cycles, where the traced signal flips from $0$ to $1$ or vice versa.
	The result is the target vector $\tvec$. 
	Moreover, the procedure records the precise number $\anum$ of changes.
	In this setting, a witness is a reconstruction of the traced signal.
	The characteristic of the problem is that the size of $\setvec$ is typically large while $\anum$ is small.
	This is due to the fact that the interval of clock cycles is comparably long, while a change in the signal only happens rarely. 

	The logging mechanism is used in failure analysis. 
	Once a failure occurred, the value of the logged target vector $\tvec$ gets stored. 
	Subsequently, all traces leading to the vector $\tvec$ have to be reconstructed. 
	This is achieved by a precise satisfiability-modulo-theories analysis or a simulation in hardware. 
	The expensive step in this analysis is to ensure completeness:
	finding (roughly) all witnesses without knowing them. 
	Providing the number of witnesses (in an approximate model), $\witcountprob$ allows us to judge the (degree of) completeness of the reconstruction. 
	
	Hence, there is a need for an algorithm that determines the number of witnesses and has a low dependence on $\sizesetvec$.
	A simple approach is a dynamic programming that assumes an order on the set $\setvec = \set{v_1, \dots, v_\sizesetvec}$ with $v_1 < \dots < v_\sizesetvec$.
	We compute a table $C[\avec, i, w]$ for $\avec \in \F^\adim$, $i \in [1,k]$, and $w \in \setvec$.
	An entry $C[\avec, i, w]$ is the number of ways to write $\avec$ as a sum of $i$ different vectors from $\setvec$, where the vectors are increasingly ordered and the last one seen is $w$.
	We have that $\power{\witgiven} = \anum! \cdot \sum_{v \in \setvec} C[\tvec, \anum, \avec]$.
	The entries can be computed by the recurrence
	\mbox{$C[\avec, i, w] = \sum_{w' < w, w' \in \setvec} C[v-w', i-1, w']$.}
	In total, the table has $2^\adim \anum \sizesetvec$ many entries.
	Computing an entry takes $\sizesetvec$ additions.
	Hence, the needed arithmetic operations to fill the table are $\bigOS{2^\adim \sizesetvec^2}$. 
	Suppose the set $\setvec$ is large, $m = \bigO{2^\adim}$.
	Then the dynamic programming roughly takes $\bigOS{2^{3\adim}}$ operations.
	Hence, even the quadratic dependence on $\sizesetvec$ in the number of operations is prohibitive.
	
	Our main result is an algorithm solving \witcountprob\ in $\bigOS{2^{\adim}}$ arithmetic operations.
	Surprisingly, the size $\sizesetvec$ of the underlying set of vectors
	does not contribute to the number of arithmetic operations at all.
	It only appears as a logarithmic factor in the runtime.
	A similar phenomenon appeared before in the context of counting covers and partitions via inclusion-exclusion~\cite{Bjorklund2009}.
	Our algorithm can also be applied to solve $\mld$, $\syndrome$, $\weightdistr$, and $\mindist$. 
	In particular, when solutions are allowed to contain less than $k$ vectors, it is sufficient to slightly change our subroutine for $\power{\candgiven}$ (see below).

	The idea behind our algorithm is to divide the task of counting witnesses, in a way that resembles inclusion-exclusion. 
	First, we count witness \emph{candidates}, $\anum$-tuples of vectors that add up to the target vector but may repeat entries:
	\begin{align*}
		\candgiven = \Set{\avecvec\in\setvec^{\anum}}{\Sigma \avec_i = \tvec}.
	\end{align*}
	Then, we determine the number of \emph{failures}, candidates that indeed repeat an entry:
	\begin{align*}
		\failgiven = \Set{\avecvec\in\setvec^{\anum}}{\Sigma \avec_i = \tvec, \avec_i=\avec_j\text{ for some }i\neq j}.
	\end{align*}
	Counting witnesses now amounts to counting candidates and counting failures:
	\begin{align}
	\power{\witgiven} = \power{\candgiven} - \power{\failgiven}.\label{Equation:Decomposition}
	\end{align}

	We address the problem of counting candidates by computing a convolution operation over $\anum$ functions. 
    Doing this efficiently requires a trick known as Walsh-Hadamard transform~\cite{Ahmed1975,Rockmore1995}.
    It turns the convolution operation into a component-wise vector product.

    We address the problem of counting failures by means of two factorizations. 
    We first factorize the set of witnesses $\avecvec=(\avec_1,\ldots \avec_{\anum})$ into usage patterns.
    A usage pattern is a partition of the set of positions $[1, \anum]$ that indicates where vectors in $\avecvec$ repeat.
    The precise vectors are abstracted away.
    In a second step, we exploit the fact that we compute over $\F$. 
    We further factorize the usage patterns according to the parity of their equivalence classes. 
    Two patterns are considered equivalent if they have the same number of classes of odd and the same number of classes of even cardinality.
    With these factorizations, the number of failures admits a recurrence to the number of witnesses for a smaller parameter $\anum'<\anum$.
    Altogether, we arrive at a dynamic programming over the parameter $\anum$ to which the candidate count contributes only an additive factor.

    We complement our algorithm by two lower bounds.
    The first one shows that \witcountprob\ cannot be solved in time $2^{\smallo{d}}$ unless $\#\ETH$ fails.
    Here, $\#\ETH$ is a counting version of the exponential-time hypothesis, a standard lower-bound assumption in fine-grained complexity~\cite{Impagliazzo2001,Dell2014}.
    The result shows that our algorithm is optimal.
    The second lower bound shows that the decision version of $\witcountprob$ does not admit a polynomial kernel.
    Both lower bounds are obtained by reductions from perfect matching in hypergraphs.
    
	\subparagraph{Related Work.}
	
	We already discussed the related work in coding theory. 
	A key tool in our algorithm is the convolution over $\F^d$.
	Convolution-based methods \cite{Fomin2010,Cygan2015} have seen great success in parameterized complexity. 
	Bj\"orklund et al.  were the first to see the potential of the subset convolution in this context~\cite{Kaski2007}. 
	They gave an $\bigOS{2^n}$-time algorithm and applied it to partitioning problems and to a variant of $\steiner$. 
	Their algorithm is based on fast M\"obius and zeta transforms.
	The computation of the latter goes back to Yates~\cite{Yates1937}. 
	In \cite{Kaski2008}, subset convolution was applied to compute the Tutte polynomial of a graph.
	Different variants of $\domset$ were solved in \cite{Rooij2009} and \cite{Telle2015} via subset convolution. 
	Moreover, the technique was used as a subroutine in Cut \& Count~\cite{Cygan2011}. 
	The same paper presents an algorithm for the \emph{XOR product}, a convolution operation over $\mathbb{F}^d_p$, for $p = 2,4$. 
	By applying the algorithm for $p = 4$, one can lift the techniques presented in this paper to instances over $\mathbb{F}^d_4$. 
	However, one cannot generalize to arbitrary $p$ since the algorithm for the XOR product would suffer from rounding errors \cite{Cygan2011}.
	An application of subset convolution in automata theory is given in~\cite{Meyer2017}.
	
	Also other transform-based methods \cite{Rockmore1995,Ahmed1975} can be found in algorithms.
	In \cite{Cygan2009}, fast Fourier transform was instantiated to derive an algorithm for the $\bandwidth$ problem.
	Based on Yates algorithm and a space-efficient zeta transform, the domatic and the chromatic number of graphs were computed~\cite{Bjorklund2008,Kaski2010}.
	In \cite{Williams2009}, the fast Walsh-Hadamard transform was applied to solve $\kpath{k}$.
	Polynomial-space algorithms for a variety of problems were constructed in \cite{Lokshtanov2010} using transforms. 
	The authors of \cite{Kaski2012} consider efficient multiplication in M\"obius algebras in general.

	The methods in this paper are furthermore inspired by the inclusion-exclusion principle.
	It was first used by Bj\"orklund et al. in \cite{Bjorklund2006}, and independently by Koivistio in \cite{Koivisto2006} for counting covers and partitions.
	Also this technique was used in various algorithms and is particularly helpful when counting solutions \cite{Bjorklund2009,Husfeldt2006,Kaski2009,Bjorklund2010,BjorklundHamil2010}.

	\section{Parameterized Complexity}
\label{Section:Prelim}

	Our goal is to identify the influence of the parameters $\adim$ and $\sizesetvec$ on the problem $\witcountprob$. 
	Parameterized complexity provides a framework to do so. 
	We introduce the basics following~\cite{Flum2006,Downey2013}. 
	A \emph{parameterized problem} is a language $L \subseteq \Sigma^* \times \N$, where $\Sigma$ is a finite alphabet.
	The problem is called \emph{fixed-parameter tractable} if there is a deterministic algorithm that decides $(x,k) \in L$ in time $f(k) \cdot \power{x}^{\bigO{1}}$.
	Here, $f$ is any computable function that only depends on the parameter $k$.
	It is common to denote the runtime of the algorithm by $\bigOS{f(k)}$ to emphasize on the dominant factor.
	The class of all fixed-parameter-tractable problems is denoted by $\FPT$.
	The precise value of $f(k)$ is crucial.
	Finding upper and lower bounds for $f(k)$ is usually referred to as \emph{fine-grained complexity}.
	
	There is also a hardness theory.
	Not all $\NP$-complete problems are likely to be $\FPT$, for instance the problem of finding a clique of size $k$.
	Despite extensive effort, no algorithm was found that solves the problem in time $\bigOS{f(k)}$.
	In fact, the problem is known to be hard for the complexity class $\W[1]$.
	All problems with this property are unlikely to be $\FPT$.

	\section{Witness Counting}
\label{Section:Algorithm}

	Our main result is an optimal algorithm for $\witcountprob$ with logarithmic dependence on $\sizesetvec$.
	\begin{theorem}
		Let $\ops = 2^{\adim}{\adim}{\anum} + \anum^4$.
		\witcountprob\ can be solved with $\bigO{\ops}$ arithmetic operations and in time $\bigO{\ops \cdot \multtime{\anum\cdot \mylog\ \sizesetvec}}$, if $\multtime{x}$ is the time to multiply two $x$-bit numbers.
	\end{theorem}
	\noindent The best known runtime for multiplication is $\multtime{x} = x \cdot \mylog \ x \cdot 2^{\bigO{\mylog^*x}}$ \cite{Furer2007}.
	In the theorem, we assume that the set $\setvec$ is given in terms of its characteristic function $\charfun{\setvec}$.
	If $V$ should have a different representation, we can compute $\charfun{\setvec}$ in $2^\adim$ look-ups without changing the statement.
	
	Actually, given $\setvec\subseteq\allvec$, $\tvec\in\allvec$, $\anum\in\N$, our algorithm determines the solution to all instances with the same set of vectors and the same target but $\anum'\leq \anum$.
	As explained in the introduction, it relies on Equation~\eqref{Equation:Decomposition} to decompose the witness-counting problem into counting candidates and failures.
	\subsection{Candidate Counting}

	We express the number of candidates in terms of a convolution of $\anum$ functions.
	The definition is as follows \cite{Ahmed1975,Cygan2011}. 
	The \emph{convolution} of two functions $f,g : \F^\adim \rightarrow \Z$ is the function $f \ast g : \F^\adim \rightarrow \Z$ which maps $x \in \F^\adim$ to the sum $\sum_{\avec_1 + \avec_2 = x} f(\avec_1) g(\avec_2)$.
	Note that the sum of $\avec_1$ and $\avec_2$ is computed over $\F^\adim$.
	The convolution operator is associative and for functions $f_1, \dots, f_\anum : \F^\adim \rightarrow \Z$, we have
	\begin{align}
		(f_1 \ast \dots \ast f_\anum)(x) = \sum_{\avec_1 + \dots + \avec_\anum = x} f_1(\avec_1) \cdots f_\anum(\avec_\anum).\label{Equation:Convolution}
	\end{align}
	For counting the candidates, we take all functions $f_1$ to $f_{\anum}$ to be the characteristic function of the given set of vectors~$\setvec$, $\charfun{\setvec}: \F^\adim \rightarrow \set{0,1}$.
	With Equation~\eqref{Equation:Convolution}, correctness is immediate.
	\begin{fact}
		$\power{\candgiven}=(\Asterisk^\anum \, \charfun{\setvec}) (\tvec)$.
	\end{fact}
	\noindent Note that the convolution cannot check whether the summed-up vectors $\avec_i$ are different.
	This is why we need the correction term $\power{\failgiven}$. 

	It remains to compute the convolution.
	To do so efficiently, we cannot just iterate over all summands in Equation~\eqref{Equation:Convolution}.
	Instead, we apply the \emph{Walsh-Hadamard transform} ($\WHT_\adim$)~\cite{Ahmed1975}.
	We defer the definition for the moment and stick with the properties. 		
	The transform turns a convolution into a pointwise product of integer vectors of size $2^{\adim}$. 
	Moreover, it is self-inverse up to a factor.
	\begin{theorem}[\cite{Ahmed1975}]
		$\WHT_\adim (f_1 \ast \dots \ast f_\anum) = \WHT_\adim(f_1) \cdots \WHT_\adim (f_\anum)$ and $\WHT_\adim (\WHT_\adim (f)) = 2^\adim \cdot f$.
	\end{theorem}
	Given the theorem, we can compute the convolution in Equation~\eqref{Equation:Convolution} via the Walsh-Hadamard transform of $\charfun{\setvec}$, followed by multiplications of vectors of size $2^{\adim}$, followed by another transform and a division. 
	An algorithm called the \emph{fast Walsh-Hadamard transform} computes the transform in $\bigO{2^\adim \adim}$ arithmetic operations~\cite{Ahmed1975,Rockmore1995}.
	Together with the multiplications, we arrive at the following complexity estimate.
	Note that it refers to $\power{\candof{\setvec}{\tvec}{i}}$ for all $i\leq k$.
	\begin{proposition}\label{Proposition:WalshHadamard}
		$\power{\candof{\setvec}{\tvec}{i}}$ for all $i\leq k$ can be computed in $\bigO{2^{\adim} \adim \anum}$ arithmetic operations.
	\end{proposition}
	The reason we determine the intermediary values $\power{\candof{\setvec}{\tvec}{i}}$ is that the overall algorithm for computing $\power{\witgiven}$ is a dynamic programming which may access them.
	To compute only the value $\power{\candgiven}$, two transforms and an iterated multiplication with $\mylog\ \anum$ operations would suffice. Hence, the parameter $\anum$ would contribute only logarithmically to the complexity.  
	Instead, we need to apply $\anum$ multiplications and $\anum$ transforms.  
	
	The Walsh-Hadamard transform~\cite{Ahmed1975} is essential in the above algorithm.
	It is based on the \emph{Hadamard Matrix} $\hadmatrix_\adim$  defined recursively by
	$\hadmatrix_0 = \begin{bmatrix} 1 \end{bmatrix}$
	and for $\ell > 0$ by
	\begin{align*}
		\hadmatrix_{\ell} =
		\begin{bmatrix}
		\hadmatrix_{\ell-1} & \hadmatrix_{\ell-1} \\
		\hadmatrix_{\ell-1} & -\hadmatrix_{\ell-1}
		\end{bmatrix}
		\in \Z^{2^\ell \times 2^\ell}.
	\end{align*}
	The \emph{Walsh-Hadamard transform} of $f : \F^\adim \rightarrow \Z$ is defined by $\WHT_\adim(f) = \hadmatrix_\adim \cdot f$.
	In this product, $f$ is seen as a vector in $\Z^{2^\adim}$.
	The \emph{fast Walsh-Hadamard transform} is a dynamic programming algorithm to compute $\WHT_\adim(f)$.
	It is based on the recursive structure of $\hadmatrix_\adim$ and takes $\bigO{2^\adim \adim}$ arithmetic operations.
	That the transform is self-inverse is based on the fact that the Hadamard Matrix has a simple inverse, $\hadmatrix_\adim^{-1} = 1/2^\adim \cdot \hadmatrix_\adim$.
	
	\subsection{Failure Counting}

	We divide the task of counting failures along the usage pattern implemented by the failures.
	Usage pattern are defined via partitions.
	An \emph{unordered partition} of a set $\aset$ is a collection \mbox{$\aupart=\set{\aset_1,\ldots, \aset_n}$} of non-empty and disjoint subsets that together cover $\aset$ in that $\bigcup \aset_i = \aset$.
	We use $\upartof{\aset}$ for the set of all unordered partitions of $\aset$.
	We call a partition \emph{trivial} if $n=\power{\aset}$, which means the classes consist of single elements.  
	In our development, we also use \emph{ordered partitions} of $\aset$, tuples $(\aset_1,\ldots, \aset_n)$ satisfying the same non-emptiness, disjointness, and covering constraints.
	
	Recall that a failure is a tuple $(\avec_1,\ldots, \avec_{\anum})$ where $\avec_i=\avec_j$ for some $i\neq j$.
	A failure induces an equivalence $\vecequivof{\avecvec}$ on the set of positions $[1, \anum]$ that tracks the usage of vectors: $i\vecequivof{\avecvec}j$ if $\avec_i=\avec_j$.
	The \emph{usage pattern} of the failure is the unordered partition of $[1,\anum]$ induced by $\vecequivof{\avecvec}$.
	We use the following function to extract the usage pattern of a failure:
	\begin{align*}
		\shape:\failgiven&\rightarrow \upartof{[1,\anum]}\\
		\avecvec&\mapsto\factorize{[1, \anum]}{\vecequivof{\avecvec}}.
	\end{align*}
	Note that by definition no failure maps to a trivial partition.
	There is at least one non-trivial class in $\factorize{[1,\anum]}{\vecequivof{\avecvec}}$.
	This explains the index in the following disjoint union:
	\begin{fact}\label{Fact:FailViaShape}
		$\failgiven = \biguplus_{\substack{\aupart\in \upartof{[1, \anum]}\\\aupart\text{ non-trivial}}}\shapeinvof{\aupart}$.
	\end{fact}
	The number of failures is thus the sum over $\power{\shapeinvof{\aupart}}$.
	There are, however, too many unordered partitions as can be iterated over efficiently.
	We factorize the set of unordered partitions, exploiting that we compute over $\F$. 
	The field of characteristic $2$ has the property that $\avec+\avec = 0$. 
	As a consequence, the vectors from $\avecvec$ whose indices belong to classes of even size do not contribute to the target vector. 
	Indeed, consider a class with an even number of indices, $\set{i_1,\ldots, i_{2l}}\in \shapeof{\avecvec}$. 
	By definition of the equivalence, $v_{i_1}=\ldots =v_{i_{2l}}$ and hence, the vectors cancel out.
	Similarly, the vectors whose indices belong to classes of odd size only contribute a single vector to the target. 
	
	The discussion motivates the following definitions.
	Given a partition $\aupart\in\upartof{[1, k]}$, we use $\evenof{\aupart}\in \N$ to denote the number of classes in $\aupart$ that have an even cardinality.
	Similarly, we let $\oddof{\aupart}\in \N$ denote the number of classes of odd cardinality in $\aupart$.
	The \emph{parity-counting function} maps a partition to this pair of values:
	\begin{align*}
		\parityfun:\upartof{[1, \anum]}&\rightarrow [1, \anum]\times[1, \anum]\\
		\aupart&\mapsto (\evenof{\aupart}, \oddof{\aupart}).
	\end{align*}
	Note that $\parityfun$ factorizes the set of partitions. 
	The following lemma is key to our development.
	It shows that, as far as counting is concerned, $\shapeinv$ is insensitive to this factorization. 
	\begin{lemma}\label{Lemma:ShapeViaParity}
		Let $\aupart\in \upartof{[1, k]}$ with $\parityfunof{\aupart}=(\evencount, \oddcount)$. 
		Recall that $\sizesetvec = \power{\setvec}$.
		We have 
		\begin{align*}
			\power{\shapeinvof{\aupart}}=
			\power{\witof{\setvec}{\tvec}{\oddcount}}\cdot\frac{(\sizesetvec-\oddcount)!}{(\sizesetvec-\oddcount-\evencount)!}.
		\end{align*}
	\end{lemma}
	\begin{proof}
		Consider $\aupart\in \upartof{[1, k]}$ with $\parityfunof{\aupart}=(\evencount, \oddcount)$. 
		In a first step, we characterize $\shapeinvof{\aupart}$ in terms of a set of functions that is easier to count. 
		The observation is this. 
		Every failure $\avecvec\in \shapeinvof{\aupart}$ is uniquely determined by choosing a vector $\avec\in \setvec$ for each class in $S\in \aupart$. 
		Formally, an \emph{instantiation} of the partition $\aupart$ is an injective function
		\begin{align*}
			\inst:\aupart\rightarrow \setvec\qquad\text{such that}\qquad
			\tvec = \sum_{\substack{S\in \aupart\\
			\power{S}\text{ odd}}}\instof{S}.
		\end{align*}
		Injectivity follows from the definition of $\vecequivof{\avecvec}$, requiring the vectors associated to different classes to be different. 
		That the sum of the vectors has to be the target vector~$\tvec$ is by the definition of failures. 
		As we compute over $\F$, only the vectors associated to classes of odd cardinality contribute to the target vector.
		We use $\setinstof{\aupart}$ to refer to the set of all instantiation functions.  
		
		In turn, every failure $\avecvec\in \shapeinvof{\aupart}$ induces an instantiation $\inst\in \setinstof{\aupart}$. 
		To a class $S$ with $i\in S$, we associate the vector $\avec_{i}$. 
		The definition of $\vecequivof{\avecvec}$ ensures independence of the representative. 
		Combined with the previous paragraph, we have a bijection between $\shapeinvof{\aupart}$ and $\setinstof{\aupart}$.
		
		By sorting the partitions, instantiations $\inst\in\setinstof{\aupart}$ can be seen as pairs \mbox{$(\avecvecp, \avecvecpp)\in \setvec^{\oddcount}\times \setvec^{\evencount}$.}
		The first component $\avecvecp$ is an $\oddcount$-tuple of different vectors from $\setvec$ that sum-up to $\tvec$. 
		Phrased differently, $\avecvecp$ is a solution to $\witof{\setvec}{\tvec}{\oddcount}$, the witness-counting problem where the number is reduced to $\oddcount$. 
		The second component $\avecvecpp$ is an $\evencount$-tuple of distinct vectors that do not occur in $\avecvecp$. 
		For the first component, there are $\power{\witof{\setvec}{\tvec}{\oddcount}}$ many choices.
		For each first component, there are \mbox{$(\sizesetvec-\oddcount)\cdots (\sizesetvec-\oddcount-\evencount + 1)=\frac{(\sizesetvec-\oddcount)!}{(\sizesetvec-\oddcount-\evencount)!}$} many choices left for the second component.
	\end{proof}
	Combining Fact~\ref{Fact:FailViaShape} and Lemma~\ref{Lemma:ShapeViaParity}, we arrive at a formula for counting failed candidates.
	The inequality $\evencount+\oddcount<\anum$ is again due to the fact that failures only induce non-trivial partitions. 
	\begin{proposition}\label{Proposition:FailuresViaWit}
		$\power{\failgiven} = \sum_{\substack{(\evencount, \oddcount)\in \N\times \N\\ \evencount+\oddcount < \anum} }\power{\parityfuninvof{\evencount, \oddcount}}\cdot\frac{(m-\oddcount)!}{(m-\oddcount-\evencount)!}\cdot \power{\witof{\setvec}{\tvec}{\oddcount}}$.
	\end{proposition}
	The proposition yields a recurrence to determine  the number of witnesses $\power{\witgiven}$. 
	We return to this in a moment when we discuss the overall algorithm. 
	A factor that is local to the failure count is $\power{\parityfuninvof{\evencount, \oddcount}}$, the number of partitions with a given parity count.
	It can be determined with a dynamic programming that runs in polynomial time.
	\begin{lemma}\label{Lemma:TableParityFun}
		Computing all $\power{\parityfuninvof{\evencount, \oddcount}}$ with $\evencount+\oddcount<\anum$ needs $\bigO{\anum^4}$ arithmetic operations.
	\end{lemma}
	\begin{proof}
		Given a $\anum$-element set, we show how to compute the number of ordered partitions with $\evencount$-many classes of even and $\oddcount$-many classes of odd cardinality. 
		Let $\parcounttab{\anum}{\evencount}{\oddcount}$ represent this number.
		The corresponding number of unordered partitions can simply be obtained by a division:
		\begin{align*}
			\power{\parityfuninvof{\evencount, \oddcount}} = \parcounttab{\anum}{\evencount}{\oddcount} / (\evencount + \oddcount)! \ .
		\end{align*}
		The number $\parcounttab{\anum}{\evencount}{\oddcount}$ satisfies the following recurrence, where $[\text{cond}]$ is a function that evaluates to $1$ if the given condition holds and to $0$ otherwise: 
		\begin{align*}
			\parcounttab{s}{x}{y} = \sum_{p=1}^s \binom{s}{p}\big([p\text{ even}]\cdot \parcounttab{s-p}{x-1}{y}+[p\text{ odd}]\cdot \parcounttab{s-p}{x}{y-1}\big).
		\end{align*}
		The equation can be understood as a recursion that explicitly builds-up a partition. 
		It adds to the current partition a new class with $p$ elements. 
		Depending on whether $p$ is even or odd, the construction continues with the parameters adjusted accordingly.  
		The bases cases are immediate. 
		
		The algorithm first tabulates the binomials using $\anum^3$ operations.
		Indeed, as $s, p\leq k$ there are quadratically many pairs each of which requiring at most $\anum$ multiplications.
		The computation is then a dynamic programming that fills a table of size $\anum^3$.
		For each entry, we sum over at most $\anum$ numbers, where we can look-up earlier computed entries and the binomial.
		Hence, we need $\bigO{\anum}$ operations for a single entry.
		This results in $\bigO{\anum^4}$ arithmetic operations to fill the whole table.
	\end{proof}

	\subsection{Overall Algorithm}

	With Equation~\eqref{Equation:Decomposition} and Proposition~\ref{Proposition:FailuresViaWit}, we obtain
	\begin{align*}
		\power{\witgiven}
		&= \power{\candgiven} - \sum_{\substack{(\evencount, \oddcount)\in \N\times \N\\ \evencount+\oddcount < \anum} }\power{\parityfuninvof{\evencount, \oddcount}}\cdot\frac{(m-\oddcount)!}{(\sizesetvec-\oddcount-\evencount)!}\cdot \power{\witof{\setvec}{\tvec}{\oddcount}}.
	\end{align*}
	The overall algorithm for computing $\power{\witgiven}$ is thus a dynamic programming over $\anum$. 
	It accesses powerful look-up tables that we initialize in a first phase. 
	We compute $\power{\candof{\setvec}{\tvec}{i}}$ for all $i\leq \anum$ using the Walsh-Hadamard transform.
	As stated in Proposition~\ref{Proposition:WalshHadamard}, this means $\bigO{2^{\adim} \adim \anum}$ arithmetic operations.
	Moreover, we tabulate $\power{\parityfuninvof{\evencount, \oddcount}}$ for all $\evencount+\oddcount<\anum$. 
	According to Lemma~\ref{Lemma:TableParityFun}, this costs $\bigO{\anum^4}$ arithmetic operations.
	We also precompute $\frac{(\sizesetvec-\oddcount)!}{(\sizesetvec-\oddcount-\evencount)!}$ for all $\oddcount+\evencount<\anum$. 
	The factorials cancel out. 
	For given $\oddcount$ and $\evencount$, we thus have at most $\anum$ multiplications. 
	Altogether, also that table can be filled with $\bigO{\anum^3}$ operations.
	
	The dynamic programming takes $\anum$ iterations. 
	In each iteration, we have a sum over at most $\anum^2$ numbers.
	Since we have tabulated for each $\evencount+\oddcount<\anum$ all the  values needed, we can evaluate the sum in $\bigO{\anum^2}$ operations.
	Hence, the overall number of operations in the dynamic programming is $\bigO{\anum^3}$. 
	Together with the initialization, this yields the announced
	$\bigO{2^{\adim} \adim \anum + \anum^4}$ operations.
	
	The numbers over which we compute, both in the initialization and in the dynamic programming, are all positive and bounded from above by $\power{\candgiven}$, the number of candidates. 
	This number is at most $\sizesetvec^{\anum}$. 
	Hence, the arithmetic operations are executed on $\mylog\ \sizesetvec^{\anum}=\anum\cdot \mylog\ \sizesetvec$-bit numbers. 
	We estimate the cost of an operation as that of a multiplication. 
	Actually, the Walsh-Hadamard transform only uses additions and subtractions and is therefore slightly cheaper.

	
	\section{Lower Bounds}
\label{Section:Lowerbound}

	We present two lower bounds for $\witcountprob$. 
	The first lower bound shows that the existence of a $2^{\smallo{d}}$-time algorithm would contradict $\# \ETH$, a version of the exponential-time hypothesis for counting problems. 
	The second lower bound shows that the decision version of $\witcountprob$ does not admit a polynomial kernel unless $\NP \subseteq \conpoly$.
	Both bounds are based on a reduction from $\perfmatch{\ell}$, the problem of determining whether an $\ell$-uniform hypergraph admits a perfect matching.
	We introduce the needed notions.
	
	A \emph{hypergraph} is a pair $\graph = (\univ,\edges)$, where $\univ$ is a finite set of \emph{vertices} and $\edges$ is a set of \emph{edges}.
	Edges in a hypergraph connect a number of vertices. 
	Formally, the set of edges is a collection of subsets of vertices, $\edges \subseteq \powerset{\univ}$. 
	A hypergraph is $\ell$\emph{-uniform} if every edge connects exactly $\ell$ vertices, $\power{e} = \ell$ for all $e\in\edges$.
	Note that a $2$-uniform hypergraph is just an undirected graph.
	
	A \emph{perfect matching} of an $\ell$-uniform hypergraph $\graph = (\univ, \edges)$ is an independent set of edges that covers $\univ$.
	To be precise, it is a subset $\per \subseteq \edges$ such that all vertices are contained in an edge in $\per$ and no two edges in $\per$ share a vertex.
	Note that a perfect matching consists of exactly $\power{\univ} / \ell$ many edges and thus only exists if $\power{\univ}$ is divisible by $\ell$.
	We use $\perfect{\graph}$ to denote the set of perfect matchings of $\graph$.
	For fixed $\ell$, the problem $\perfmatch{\ell}$ is the following: 
	given an $\ell$-uniform hypergraph $\graph$, decide whether there exists a perfect matching of $\graph$.
	
	We show a polynomial-time reduction from the $\perfmatch{\ell}$ problem to the decision variant of $\witcountprob$. 
	It works for any $\ell$.
	The reduction implies a complexity and a kernel lower bound. 
	For the complexity lower bound, we establish a relationship between the number of perfect matchings in an undirected graph and the number of witnesses.
	Since perfect matchings cannot be counted in $2^{\smallo{\power{U}}}$ assuming $\# \ETH$, we obtain a lower bound for $\witcountprob$. 
	For the kernel lower bound, we use the fact that $\perfmatch{\ell}$ does not admit a polynomial compression of a certain size. 
	For both results, it is important that the reduction yields a parameter $\adim$, the dimension in $\witcountprob$, linear in the size of the given vertex set $\power{\univ}$.
	\begin{lemma}\label{Lemma:Polyreduction}
		For any $\ell \geq 2$, there is a polynomial-time reduction from $\perfmatch{\ell}$ to the decision version of $\witcountprob$.
		Moreover, the parameter $\adim$ is linear in $\power{\univ}$.
	\end{lemma}
	For the construction, let $\graph = (\univ,\edges)$ be an $\ell$-uniform hypergraph.
	Assume the vertices in $\univ$ are ordered.
	We construct an instance $(\setvec, \tvec, \anum)$ of $\witcountprob$ over $\F^\adim$, where $\adim = \power{\univ}$.
	To this end, let $e \in \edges$.
	We define $\lambda(e)\in \F^\adim$ to be the bitvector representation of $e$, with $\lambda(e)(u) = 1$ if and only if $u \in e$.
	All these vectors are collected in the set $\setvec = \Set{ \lambda(e) }{ e \in \edges}$.
	Since we need to cover all vertices in $\graph$, we set $\tvec = (1, \dots, 1) \in \F^\adim$.
	The trick is to define $\anum$, the number of vectors to select, such that the corresponding edges are forced to be pairwise disjoint. 
	We set $\anum = \adim / \ell$.
	If this is no natural number, we clearly have a no-instance.
	With $\ell$ entries $1$ per vector, the only way to cover $\tvec$ by $\anum$ vectors is to avoid overlapping entries.
	The following lemma states correctness of the reduction.  
	It actually shows that the reduction is parsimonious up to a factorial.
	The factorial appears since witnesses are ordered while perfect matchings are not.
	\begin{lemma}\label{Lemma:Perfectcount}
		$\power{\perfect{\graph}} = \anum! \cdot \power{\witgiven}$.
	\end{lemma}
	\begin{proof}
		We will reuse the map $\lambda$ from above.
		Note that it is a bijection between $\powerset{\univ}$ and $\F^\adim$.
		Let $\per = \set{e_1, \dots, e_\anum}$ be a perfect matching of $\graph$.
		Since $e_i \cap e_j = \emptyset$ for $i \neq j$, we have \mbox{$\lambda(e_i) + \lambda(e_j) = \lambda(e_i \cup e_j)$.}
		Hence, $\sum_{i \in [1,\anum]} \lambda(e_i) = \lambda(\bigcup_{i \in [1,\anum]} e_i) = \lambda(\univ)$.
		The vector $\lambda(\univ)$ is the target~$\tvec$, which means $(\lambda(e_1), \dots, \lambda(e_\anum))$ is a witness.
		Of course, any reordering of $(\lambda(e_1), \dots, \lambda(e_\anum))$ is a witness as well.
		Hence, any perfect matching yields the existence of $\anum!$ witnesses.
		
		Now let $\avecvec = (\avec_1, \dots, \avec_\anum) \in \setvec^\anum$ be a witness.
		Since $\avec_i \in \setvec$, the set $e_i = \lambda^{-1}(\avec_i)$ is an edge in $\graph$.
		We define $\per = \set{e_1, \dots, e_\anum} \subseteq \powerset{\univ}$. 
		This is a perfect matching.
		That all vertices are covered is by the choice of $\tvec$. 
		Disjointness of the edges follows from the choice of $\anum$.
	\end{proof}
	\subsection{Lower Bound on the Runtime}

	We prove the announced lower bound on the runtime for $\witcountprob$.
	It shows that the algorithm presented in this paper is optimal.
	The bound is based on the $\# \ETH$, introduced in~\cite{Dell2014}.
	This variant of the exponential-time hypothesis assumes that the number of satisfying assignments of a $3$-CNF formula over $n$ variables cannot be counted in time $2^{\smallo{n}}$.
	\begin{theorem}
		$\witcountprob$ does not admit an $2^{\smallo{\adim}}$-time algorithm, unless $\#\ETH$ fails.
	\end{theorem}
	Our theorem is based on a result shown by Curticapean in \cite{Curticapean2015}.
	There, $\# \ETH$ was used to prove the existence of a $2^{\smallo{\power{\univ}}}$-time algorithm for counting perfect matchings in undirected graphs highly unlikely.
	Since undirected graphs are $2$-uniform hypergraphs, the problem corresponds to the counting variant of $\perfmatch{2}$ which we will denote by $\# \perfmatch{2}$.
	\begin{theorem}[\cite{Curticapean2015}]
		$\# \perfmatch{2}$ cannot be solved in $2^{\smallo{\power{\univ}}}$ time, unless $\# \ETH$ fails.
	\end{theorem}
	Suppose we had an algorithm for $\witcountprob$ with runtime $2^{\smallo{\adim}}$.
	Given a graph $\graph = (\univ, \edges)$, we could apply the reduction from Lemma \ref{Lemma:Polyreduction} to get, in polynomial time, an instance $(\setvec, \tvec, \anum)$ of $\witcountprob$ with $\adim = \power{\univ}$.
	Then we could apply our algorithm for counting witnesses.
	With Lemma \ref{Lemma:Perfectcount}, this yields a $2^{\smallo{\power{\univ}}}$-time algorithm for $\# \perfmatch{2}$.
	
	\subsection{Lower Bound on the Kernel}
	
	We present the lower bound on the kernel size for the decision variant of $\witcountprob$.
	From an algorithmic point of view, this is interesting since kernels characterize the number of hard instances of a problem.
	In fact, it can be shown that a kernel for a problem exists if and only if the problem is $\FPT$~\cite{Cygan2015}.
	For many problems, the search for small kernels is ongoing research, but not all problems admit such.
	This led to an approach that tries to disprove the existence of kernels of a certain size, see \cite{Bodlaender2014, Kratsch2014, Bodlaender2009, Fortnow2011}.
	For the next theorem, we apply techniques developed in that line of work.
	\begin{theorem}\label{Theorem:Kernelsize}
		Deciding $\power{\witgiven} > 0$ does not admit a poly. kernel unless $\NP \subseteq \conpoly$.
	\end{theorem} 
	A key tool in the search for kernel lower bounds are polynomial compressions.
	Let $Q$ be a parameterized and $L$ be any unparameterized problem.
	A \emph{polynomial compression} of $Q$ into $L$ is an algorithm that takes an instance $(x,k)$ of $Q$, runs in polynomial time in $x$ and $k$, and returns an instance $y$ of $L$ such that:
	(1) $y \in L$ if and only if $(x,k) \in Q$, and
	(2) $\power{y} \leq p(k)$ for a polynomial~$p$.
	The polynomial $p$ is referred to as the \emph{size} of the compression.
	
	A \emph{kernelization} is similar to a compression with two differences:
	it maps to itself and is more relaxed on the size.
	The former means that instances of $Q$ always get mapped to instances of $Q$. 
	The latter means that the size $p$ is not restricted to be a polynomial.
	We allow for any computable function which depends only on the parameter $k$.
	In the special case where $p$ is a polynomial, we call the kernelization a \emph{polynomial kernel}.
	If a polynomial compression or a kernelization exists, we say that $Q$ \emph{admits} a polynomial compression/kernelization.
	
	The proof of Theorem \ref{Theorem:Kernelsize} uses a lower bound on polynomial compressions for $\perfmatch{\ell}$.
	We combine the result with the reduction from Lemma \ref{Lemma:Polyreduction} and derive the wanted kernel lower bound for the decision variant of $\witcountprob$.
	\begin{theorem}[\cite{Dell2012,Cygan2015}]\label{Theorem:KernelPerfMatch}
		Let $\varepsilon > 0$.
		For any $\ell \geq 3$, $\perfmatch{\ell}$ parameterized by $\power{\univ}$ does not admit a polynomial compression of size 
		$\bigO{\power{\univ}^{\ell - \varepsilon}}$, unless $\NP \subseteq \conpoly$ .
	\end{theorem} 
	For the proof of Theorem \ref{Theorem:Kernelsize}, assume there is a polynomial kernel for deciding $\power{\witgiven} > 0$.
	This means there is an algorithm $\alg$ that takes an instance $I = (\setvec, \tvec, \anum)$ over $\F^\adim$ and maps it to an instance $I' = (\setvec',\tvec',\anum')$.
	The algorithm $\alg$ runs in polynomial time in $\power{I}$ and~$\adim$, and the size of $I'$ is bounded by a polynomial in $\adim$:
	$\power{I'} = \bigO{\adim^{\ell}}$, for a constant $\ell$.
	
	To derive a contradiction, consider the problem $\perfmatch{(\ell+1)}$.
	Let $\graph$ be an input to the problem.
	By Lemma \ref{Lemma:Polyreduction}, we get, in polynomial time, an instance $I = (\setvec, \tvec, \anum)$ of the decision variant of $\witcountprob$ with parameter $\adim = \power{\univ}$.
	If we apply algorithm $\alg$ to the instance, we get an instance $I' = (\setvec',\tvec',\anum')$.
	As mentioned above, the size of the instance is bounded by $\bigO{\adim^\ell}$.
	
	Putting things together, we get a polynomial compression for $\perfmatch{(\ell+1)}$ of size $\bigO{\adim^\ell} = \bigO{\power{\univ}^\ell}$.
	But this contradicts Theorem \ref{Theorem:KernelPerfMatch} and concludes the proof.
	

	\section{Conclusion}
\label{Section:Conclusion}

	We studied the witness-counting problem: 
	given a set of vectors $\setvec$ in $\F^\adim$, a target vector $\tvec \in \F^\adim$, and an integer $\anum \in \N$, the task is to count all ways in which $\anum$ different vectors from $\setvec$ can be summed-up to the target vector $\tvec$.
	The problem generalizes fundamental questions from coding theory and has applications in hardware monitoring.
	
	Our contribution is an efficient algorithm that runs in time  $\bigOS{2^{\adim}}$. 
	Crucially, it only has a logarithmic dependence on the  number of vectors in $V$.
	On a high-level, the algorithm can be understood as a convolution the precision of which is improved by means of inclusion-exclusion-like correction terms --- an approach that may have applications beyond this paper.
	The algorithm as it is will generalize to vectors over $\mathbb{F}_4$ but beyond will face rounding errors.
	
	We also showed optimality: there is no algorithm solving the problem in time $2^{\smallo{\adim}}$ unless $\#\ETH $ fails. 
	Furthermore, the problem of checking the existence of a witness does not admit a polynomial kernel, unless $\NP \subseteq \conpoly$.
	Both lower bounds rely on a reduction from $\perfmatch{\ell}$, the problem of finding a perfect matching in an $\ell$-uniform hypergraph.

	\bibliographystyle{plain}
	\bibliography{content/cite}

\end{document}